\RequirePackage{ifpdf}
\ifpdf 
\documentclass[pdftex]{sigma}
\else
\documentclass{sigma}
\fi

\numberwithin{equation}{section}

\numberwithin{theorem}{section}
\numberwithin{lemma}{section}

\begin{document}

\renewcommand{\PaperNumber}{***}

\FirstPageHeading

\ShortArticleName{New Fourier Transform and Their Applications}

\ArticleName{New Fourier Transform Containing a Pair of Complex Euler Gamma Functions With a Monomial: Mathematical and Physical Applications}




\Author{Sid-Ahmed YAHIAOUI, Othmane CHERROUD and Mustapha BENTAIBA}
\AuthorNameForHeading{S.-A. Yahiaoui, O. Cherroud and M. Bentaiba}
\Address{LPTHIRM, D\'epartement de physique, Facult\'e des sciences, Universit\'e Sa\^ad DAHLAB-Blida~1, B.P.~270 Route de Soum\^aa, 09\:000 Blida, Algeria}
\Email{\href{mailto:email@address}{s$\_$yahiaoui@univ-blida.dz}, \href{mailto:email@address}{Cherroud.Othmane@univ-medea.dz}, \href{mailto:email@address}{bentaiba@univ-blida.dz}}


\Abstract{One of the goals of the present paper is to propose an elementary method to find a general formula for a new Fourier transform containing a pair of complex Euler gamma functions with a monomial $s^m$ in terms of Gauss's hypergeometric functions $_2F_1$. We further collect some mathematical results that follow by means of this transform. Physical applications in quantum mechanics requiring the expectation values of position and momentum operators for a quantum system endowed with position-dependent effective mass are presented.}


\Classification{42A38; 81S30; 58D30; 11B68}

\section{Introduction}%

\noindent The gamma function, introduced by Euler in the eighteenth century, is constantly found in diverse branches of mathematics and physics. The reason is that it has extremely varied properties and its importance lies from its usefulness in connecting to other different functions, such as Bessel and Legendre functions, (confluent-) hypergeometric functions and orthogonal polynomials, (see for example \cite{1}). Since Euler, almost every mathematician had brought a personal contribution to study such a function and till today all our knowledge on this subject is particularly rich and abundant \cite{2}.\\
\indent In this context, we were faced in our paper \cite{3} with an integral of the type
\begin{eqnarray}\label{1.1}
  \mathcal F^{(\alpha,\beta)}_m(0) &=& \int_{-\infty}^{+\infty} s^m\Gamma(\alpha-is)\Gamma(\beta+is)\,ds,
  \qquad(\alpha,\beta>0,\,s\in\mathbb R)
\end{eqnarray}
where $m=0,1,2,\cdots$, in order to evaluate the expectation values of position and momentum operators for a quantum system endowed with position-dependent effective mass, using to this end the Wigner's function in conjunction with the Weyl's transform in the framework of the phase-space quantum mechanics.\\
\indent However, the first author of this paper was surprised to see that the integral in \eqref{1.1}, as well as its more general version containing a pair of complex Euler gamma functions with weight measure $\{e^{-i\lambda s}s^m\}_{m\in\mathbb N,\,\lambda\in\mathbb R}$, namely
\begin{eqnarray}\label{1.2}
  \mathcal F^{(\alpha,\beta)}_m(\lambda) &=& \int_{-\infty}^{+\infty}e^{-i\lambda s} s^m\Gamma(\alpha-is)\Gamma(\beta+is)\,ds,\qquad(\lambda\in\mathbb R)
\end{eqnarray}
does not have known precise value, and thus does not appear anywhere in the different handbooks and monographs such as those of Gradshteyn and Ryzhik \cite{4}, Jeffrey \cite{5} and Brychkov \cite{6}. Even software systems able to perform mathematical operations such as \textsc{Mathematica} and \textsc{Maple} do not give any results. In sight of these observations, we have believed that the main reason of this omission was that \eqref{1.2} {\it does not absolutely converge} as required by the Fourier transform, since indeed the integral \eqref{1.2} can be seen as a Fourier transform of the function $f_m^{(\alpha,\beta)}(s)=s^m\Gamma(\alpha-is)\Gamma(\beta+is)$,
\begin{eqnarray}\label{1.3}
  \mathfrak F\left[f_m^{(\alpha,\beta)}(s)\right](\lambda) &\equiv& \mathcal F^{(\alpha,\beta)}_m(\lambda) =
                   \int_{-\infty}^{+\infty}e^{-i\lambda s}\,f_m^{(\alpha,\beta)}(s)\,ds,
\end{eqnarray}
as well as the $m$th derivative, with respect to the parameter $\lambda$, of the Fourier transform of the function $g^{(\alpha,\beta)}(s)=\Gamma(\alpha-is)\Gamma(\beta+is)$,
\begin{eqnarray}\label{1.4}
  \frac{d^m}{d\lambda^m}\,\mathfrak F\left[g^{(\alpha,\beta)}(s)\right](\lambda) &=& (-i)^m\,\mathfrak F\left[s^m g^{(\alpha,\beta)}(s)\right](\lambda).
\end{eqnarray}
\indent However, as is well known, there are functions that are not absolutely integrable (i.e., convergent) but have a Fourier transform satisfying \eqref{1.3}. One important example is given by the function $\sin(\lambda_0s)/s$. It is the purpose of this paper to propose a method of derivation in order to deduce a general formula for integrals in \eqref{1.2} in terms of Gauss's hypergeometric function $_2F_1$ that reduces to a polynomial of order $c=-M$, ($M\in\mathbb N$) which, as far as we know, have not been evaluated yet. The result we obtain was presented elsewhere in the appendix B of \cite{3}, while in this paper we provide a setting within which the general solution is straightforwardly derived using the Fa\`a di Bruno formula.
\paragraph{1.1. Fourier transform.} The Fourier transform is certainly one of the most useful of the unlimited number of possible integral transforms \cite{7}. Its importance has been enhanced by the development of generalizations extending the set of functions that can be Fourier transformable. For instance, the above original function $f_m^{(\alpha,\beta)}(s)\in L^1(\mathbb R)$ is Fourier transformable if and only if the integral \eqref{1.3} exists \cite{7,8}.\\
\indent Then we say that $\mathcal F_m^{(\alpha,\beta)}(\lambda)$ is:
\begin{enumerate}
  \item bounded,
  \item continuous on $\mathbb R$, and
  \item tends to zero for $\lambda\rightarrow\pm\infty$, (i.e., $\lim_{\lambda\rightarrow\pm\infty}\mathcal F_m^{(\alpha,\beta)}(\lambda)=0$).
\end{enumerate}

\indent The existence and bounded-ness of $\mathcal F_m^{(\alpha,\beta)}(\lambda)$ follow from the inequality
\begin{eqnarray}\label{1.5}
  |\mathcal F_m^{(\alpha,\beta)}(\lambda)| &\leq& \int_\mathbb R|f_m^{(\alpha,\beta)}(s)|\,ds,
\end{eqnarray}
and is a sufficient condition for that the restrictions (2) and (3) hold.
\paragraph{1.2. Fa\`a di Bruno's formula.} The Fa\`a di Bruno formula is applied to a composite function $f\circ g(x)=f(g(x))$ in order to compute its $m$th derivative in terms of the derivatives of $f$ and $g$, and may be written in the form (see, \cite{9} for an excellent survey and references therein)
\begin{equation}\label{1.6}
  \frac{d^m}{dx^m}\,f\circ g(x) = \sum_{M\in\sigma}\frac{m!}{i_1!\,i_2!\cdots i_m!}\cdot
  \left(\frac{d^M f}{dg^M}\circ g(x)\right)\cdot\prod_{j=1}^m\left(\frac{g^{(j)}(x)}{j!}\right)^{i_j},
\end{equation}
where the summation ranges over the set $\sigma$ of all partitions of $m$, that is over all different solutions in non negative integers $(i_1,i_2,\cdots,i_m)$ of the equations $i_1+2i_2+\cdots+mi_m=m$ and $i_1+i_2+\cdots+i_m=M$, with $M=1,2,\cdots,m$.\\
\indent For instance, for $m=3$ there are three possible sequences of such integers: $(i_1,i_2,i_3)=\{(0,0,1),(1,1,0),(3,0,0)\}$, while the case for $m=4$ can be partitioned in five distinct ways: $(i_1,i_2,i_3,i_4)=\{(0,0,0,1),(1,0,1,0),(0,2,0,0),(2,1,0,0),(4,0,0,0)\}$ (see Table~1 below). In number theory and combinatorics \cite{10}, the solution $(i_1,i_2,\cdots,i_m)$ is called a partition of $m$ with $M$ parts. For $1\leq M\leq m$, the set of all partition $m$ into $M$ is denoted by the subset $\sigma(M)$, where $\sigma=\bigcup_{M=1}^m\sigma(M)$.
\paragraph{1.3. Outline.} The remainder of the paper is as follows. In section 2, we establish a new general Fourier integral formula relating large classes of integrals. The focus of sections 3 and 4 is on the new results that follow from the general integral for special values of the parameters. Further results are examined, amongst them we can recover some old-known results and others are new. For example, we succeed to obtain another and new representation of the Bernoulli and Euler numbers and polynomials different from those known in the literature. We also established the expressions for monomials $\beta^m$, residues of the gamma function and the associated Laguerre polynomials in terms of the discrete summation of gamma functions. In section 5, in connection with our paper \cite{3}, we present a direct application involving the evaluation of expectation values of position and momentum operators. Finally, as is custom, the last section will be devoted to our conclusion.

\section{New general Fourier integral}%

\noindent Before we state our main result, we first investigate the possibility of convergence of \eqref{1.2}. One immediate property of the gamma function is in the use of the relation (see, for example exercises 13.1.16 and 13.1.18 in \cite{1})
\begin{eqnarray}\label{2.1}
  |\Gamma(\xi\pm i\eta)| &=&
  \Gamma(\xi)\prod_{p=0}^\infty\frac{1}{\sqrt{1+\frac{\eta^2}{(p+\xi)^2}}} \leq \Gamma(\xi),
\end{eqnarray}
applied for all values $\Re(\xi)>0$ and $\eta\in\mathbb R$. Inserting the inequality in \eqref{2.1} into \eqref{1.5}, we get
\begin{eqnarray}\label{2.2}
  |\mathcal F^{(\alpha,\beta)}_m(\lambda)|
     &=&    \left|\int_{-\infty}^{+\infty}e^{-i\lambda s} s^m\Gamma(\alpha-is)\Gamma(\beta+is)\,ds\right| \nonumber \\
     &\leq& \int_{-\infty}^{+\infty}\left|e^{-i\lambda s} s^m\Gamma(\alpha-is)\Gamma(\beta+is)\right|ds   \nonumber \\
     &\leq& 2\Gamma(\alpha)\Gamma(\beta)\int_{0}^{+\infty}s^m\,d s,
\end{eqnarray}
where the integral considered in the r.h.s of \eqref{2.2} is certainly divergent for all $m=0,1,2,\cdots$, and thus $|\mathcal F^{(\alpha,\beta)}_m(\lambda)|$ is manifestly bounded by \eqref{1.5}. We emphasize that \eqref{1.5} is sufficient and necessary to claim that the function $f_m^{(\alpha,\beta)}(s)=s^m\Gamma(\alpha-i s)\Gamma(\beta+i s)$ has a Fourier transform, under the conditions that $f_m^{(\alpha,\beta)}(s)$ is analytic and does not have any poles on the real axis.\\
\indent In the sequel, we present our general theorem concerning the computation of \eqref{1.2} which can be regarded as mostly new.
\newtheorem{Theorem}{Theorem}%
\begin{Theorem}
Let $f_m^{(\alpha,\beta)}(s)=s^m\Gamma(\alpha-is)\Gamma(\beta+is)$ denotes an analytic function in $s\in\mathbb R$, with $\alpha>0$ and $\beta>0$, and let $\mathcal F^{(\alpha,\beta)}_m(\lambda)$ be a Fourier transform of the function $f_m^{(\alpha,\beta)}(s)$ defined by \eqref{1.2}. Then
\begin{eqnarray}\label{2.3}
\mathcal F^{(\alpha,\beta)}_m(\lambda)
   &=& \int_{-\infty}^{+\infty}e^{-i\lambda s} s^m\Gamma(\alpha-is)\Gamma(\beta+is)\,ds \nonumber \\
   &=& 2\pi\,(-i)^m\,m!\,\frac{\Gamma(\alpha+\beta)}{\Gamma(\beta)}\frac{\Lambda^\beta}{(1+\Lambda)^{\alpha+\beta}} \nonumber \\
   && \times\sum_{M\in\sigma}\frac{(-1)^M\,\Gamma(\beta+M)}{\prod_{\nu=1}^m i_\nu!\,(\nu!)^{i_\nu}}\,{_2}F_1\left(-M,\alpha+\beta;\beta;\frac{\Lambda}{1+\Lambda}\right),
\end{eqnarray}
where $\Lambda=e^\lambda>0$ and here $i_k$, $m$ and $M$ are all non negative integers. The summation in \eqref{2.3} includes all different combinations of $i_k$ and satisfies simultaneously the equations
\begin{eqnarray}
  m &=& i_1+2i_2+3i_3+\cdots+mi_m, \label{2.4}\\
  M &=& i_1+i_2+i_3+\cdots+i_m. \label{2.5}
\end{eqnarray}
\end{Theorem}

\indent Before indicating the proof of this theorem, we report below in Table~1 different combinations of $i_k$, for even and odd $m$ up to 5, over all solutions in non negative integers of the equations \eqref{2.4} and \eqref{2.5}.

\begin{table}[h]
\caption{Different combinations of $i_k$ ($k=\overline{1,5}$), for even and odd $m$ up to 5, over all solutions in non negative integers of the equations \eqref{2.4} and \eqref{2.5}. By convention, we assume that for $m=M=0$ all $i_k=0$, $k=1,2,\cdots,m$.}\label{table1}
\centering
\begin{tabular}{@{}*{19}{c}}
\hline\hline\noalign{\smallskip}
$m$&\multicolumn{1}{l}{1}&\multicolumn{2}{l}{2}&\multicolumn{3}{l}{3}&\multicolumn{5}{l}{4}&\multicolumn{7}{l}{5}\\
\hline
$i_1$  &1 &2&0 &3&1&0 &4&2&1&0&0 &5&3&2&1&1&0&0 \\
$i_2$  &  &0&1 &0&1&0 &0&1&0&2&0 &0&1&0&2&0&1&0 \\
$i_3$  &  & &  &0&0&1 &0&0&1&0&0 &0&0&1&0&0&1&0 \\
$i_4$  &  & &  & & &  &0&0&0&0&1 &0&0&0&0&1&0&0 \\
$i_5$  &  & &  & & &  & & & & &  &0&0&0&0&0&0&1 \\
\hline
$M $   &1 &2&1 &3&2&1 &4&3&2&2&1 &5&4&3&3&2&2&1 \\
\noalign{\smallskip}\hline\hline
\end{tabular}
\end{table}

\begin{proof}
We start by inserting the convenient Euler's definition of the gamma function \cite{1}
\begin{eqnarray*}
  \Gamma(\alpha-is) &=& \int_0^{+\infty} e^{-t}t^{\alpha-is-1}\,dt,\quad(\alpha>0,\,s\in\mathbb R)\\
  \Gamma(\beta+is)  &=& \int_0^{+\infty} e^{-r}r^{\beta+is-1}\,dr,\quad(\beta>0,\,s\in\mathbb R)
\end{eqnarray*}
into \eqref{1.2} and interchanging between $s$- and ($r,t$)-integrations, we get
\begin{equation}\label{2.6}
\mathcal F^{(\alpha,\beta)}_m(\lambda) = \int_0^{+\infty} e^{-t}t^{\alpha-1}
  \left[\int_0^{+\infty} e^{-r}r^{\beta-1}
  \left(\int_{-\infty}^{+\infty} s^me^{i s\rho(r)}d s\right)dr\right]dt,
\end{equation}
where $\rho(r)=\ln r-\ln t-\lambda$, with $r,t\in\left]0,+\infty\right[$. Using the integral representation of the Dirac delta function, namely
\begin{eqnarray}\label{2.7}
  2\pi\delta[\rho(r)] = \int_{-\infty}^{+\infty}e^{i s\rho(r)}ds,
\end{eqnarray}
for the function $\rho(r)$, and if we likewise take the $m$th derivative of \eqref{2.7} with respect to $\rho(r)$, apparently the quantity in parenthesis in \eqref{2.6} behaves as the $m$th derivative of a delta function, i.e.,
\begin{eqnarray*}
  \int_{-\infty}^{+\infty} s^m e^{i s\rho(r)}ds = 2\pi\,(-i)^m\,\frac{d^m}{d\rho^m}\,\delta[\rho(r)],
\end{eqnarray*}
so, \eqref{2.6} becomes
\begin{equation}\label{2.8}
  \mathcal F^{(\alpha,\beta)}_m(\lambda) = 2\pi(-i)^m\int_0^{+\infty} e^{-t}t^{\alpha-1}
  \left(\int_0^{+\infty} e^{-r}r^{\beta-1}\frac{d^m}{d\rho^m}\,\delta[\rho(r)]dr\right)dt.
\end{equation}
\indent We next evaluate the $r$-integration in \eqref{2.8} using Fa\`a di Bruno's formula \eqref{1.6} (see also, \textbf{0.430} (2) in \cite{4}), taking into account that the Dirac delta function is a composite function with a real function $\rho(r)$, namely
\begin{eqnarray}\label{2.9}
  \frac{d^m}{d\rho^m}\,\delta\circ\phi(\rho) = \sum_{M\in\sigma}\frac{m!}{i_1!\,i_2!\cdots i_m!}\cdot
  \left(\frac{d^M\delta}{dr^M}\circ\phi(r)\right)\cdot\prod_{j=1}^m\left(\frac{r^{(j)}(\rho)}{j!}\right)^{i_j},
\end{eqnarray}
where the summation is taken over all solutions in non negative integers $m$ and $M$ satisfying \eqref{2.4} and \eqref{2.5}. Here $\phi(\rho)$ is the inverse of the function $\rho(r)$; i.e., $r\equiv\phi(\rho)=t\,e^{\lambda+\rho(r)}$.\\
\indent Essentially all we have to do is to take the $m$th derivative of $r=\phi(\rho)$ with respect to $\rho$, i.e.,
\[
\frac{dr}{d\rho}=\frac{d^2r}{d\rho^2}=\frac{d^3r}{d\rho^3}=\cdots=\frac{d^mr}{d\rho^m}=t\,e^{\lambda+\rho(r)}\equiv r(\rho),
\]
so that the Fa\`a di Bruno formula \eqref{2.9} is transformed into
\[
\frac{d^m}{d\rho^m(r)}\,\delta[\rho(r)] = \sum_{M\in\sigma}\frac{m!\,r^M}{\prod_{\nu=1}^m i_\nu!\,(\nu!)^{i_\nu}}
  \frac{d^M}{d r^M}\,\delta[\rho(r)],
\]
and after interchanging the order of the summation and integration, \eqref{2.8} becomes
\begin{eqnarray}\label{2.10}
\mathcal F^{(\alpha,\beta)}_m(\lambda) &=& 2\pi(-i)^m m!\sum_{M\in\sigma}\frac{1}{\prod_{\nu=1}^m i_\nu!\,(\nu!)^{i_\nu}} \nonumber \\
  &&\times\int_0^{+\infty} e^{-t}t^{\alpha-1}
  \left(\int_0^{+\infty} e^{-r}r^{\beta+M-1}\frac{d^M}{d r^M}\,\delta[\rho(r)]\,dr\right)dt.
\end{eqnarray}
\indent We shall first look  at the $r$-integration in \eqref{2.10}. To this end, we need the use of another most elementary property of the delta function involving the decomposition with a real function $\rho(r)$, namely
\[
\delta[\rho(r)]=\sum_k\frac{1}{|\rho'(r_k)|}\,\delta(r-r_k),
\]
where $r_k$ are the simple roots (zeros) of the real function $\rho(r)$; i.e., $\rho(r_k)=0$ and therefore $\rho'(r_k)\neq0$. Solving the equation $\rho(r_k)=0$ and extracting the roots, we find that $r_0=\Lambda t$, with $\Lambda=e^\lambda>0$, is the {\it unique simple root}. Then the delta function can be expressed as
\[
\delta[\rho(r)]=\Lambda t\,\delta(r-\Lambda t),\qquad(\Lambda=e^\lambda).
\]

\indent Substituting the last result into \eqref{2.10}, we find that
\begin{eqnarray}\label{2.11}
\mathcal F^{(\alpha,\beta)}_m(\lambda) &=& 2\pi\Lambda(-i)^m m!\sum_{M\in\sigma}\frac{1}{\prod_{\nu=1}^m i_\nu!\,(\nu!)^{i_\nu}} \nonumber \\
  &&\times\int_0^{+\infty} e^{-t}t^{\alpha}
  \left(\int_0^{+\infty} e^{-r}r^{\beta+M-1}\frac{d^M}{dr^M}\,\delta(r-\Lambda t)\,dr\right)dt.
\end{eqnarray}
\indent At this stage, it is useful to introduce the property
\[
\int F(x)\frac{d^M}{dx^M}\,\delta(x-x_0)\,dx=(-1)^M\frac{d^M}{d x^M}\,F(x)\Bigg|_{x\equiv x_0},
\]
in order to carry out the $r$-integration in \eqref{2.11}, then we obtain
\begin{equation}\label{2.12}
   \int_0^{+\infty} e^{-r}r^{\beta+M-1}\frac{d^M}{d r^M}\,\delta(r-\Lambda t)\,d r
   = (-1)^M\Lambda^{\beta-1}\frac{d^M}{d t^M}\,e^{-\Lambda t}t^{\beta+M-1},
\end{equation}
and \eqref{2.11} is now expressed in terms of $t$-integration as
\begin{equation}\label{2.13}
\begin{split}
\mathcal F^{(\alpha,\beta)}_m(\lambda) = 2\pi\Lambda^\beta(-i)^m m!\sum_{M\in\sigma}\frac{(-1)^M}{\prod_{\nu=1}^m i_\nu!\,(\nu!)^{i_\nu}}  \int_0^{+\infty} e^{-t}t^{\alpha}\frac{d^M}{d t^M}\,e^{-\Lambda t}t^{\beta+M-1}\,d t.
\end{split}
\end{equation}
\indent Expanding $e^{-\Lambda t}$ in its Taylor series and performing the $M$th derivative of the function $e^{-\Lambda t}t^{\beta+M-1}$ with respect to the variable $t$, we get
\begin{eqnarray}\label{2.14}
  \frac{d^M}{dt^M}\,e^{-\Lambda t}t^{\beta+M-1}
    &=& \sum_{n=0}^{+\infty}\frac{(-\Lambda)^n}{n!}M!\,{\beta+n+M-1\choose M}\,t^{\beta+n-1}\nonumber \\
    &=& t^{\beta-1}\sum_{n=0}^{+\infty}\frac{\Gamma(\beta+n+M)}{\Gamma(\beta+n)}\frac{(-\Lambda t)^n}{n!}\nonumber \\
    &\stackrel{\mbox{\footnotesize (i)}}{=}& t^{\beta-1}\frac{\Gamma(\beta+M)}{\Gamma(\beta)}\sum_{n=0}^{+\infty}
                                                  \frac{(\beta+M)_n}{(\beta)_n} \frac{(-\Lambda t)^n}{n!}\nonumber  \\
    &\stackrel{\mbox{\footnotesize (ii)}}{=}& t^{\beta-1}\frac{\Gamma(\beta+M)}{\Gamma(\beta)}\,{_1}F_1(\beta+M;\beta;-\Lambda t),
\end{eqnarray}
where (i), in \eqref{2.14}, comes from the often-used definition of Pochhammer's symbol $(a)_n=\Gamma(a+n)/\Gamma(a)$ and (ii) therein follows from the series representation expansion for the confluent hypergeometric functions \cite{1,4}. Therefore we may now write \eqref{2.13} as
\begin{eqnarray}\label{2.15}
\mathcal F^{(\alpha,\beta)}_m(\lambda) &=& 2\pi(-i)^m m!\,\frac{\Lambda^\beta}{\Gamma(\beta)}\sum_{M\in\sigma}\frac{(-1)^M\,\Gamma(\beta+M)}{\prod_{\nu=1}^m i_\nu!\,(\nu!)^{i_\nu}} \nonumber \\
  &&\times\int_0^{+\infty} e^{-t}t^{\alpha+\beta-1}\,{_1}F_1(\beta+M;\beta;-\Lambda t)\,dt.
\end{eqnarray}
\indent We can find in \cite{4} and \cite{11} many integrals involving confluent hypergeometric functions with weight measure $e^{-t}t^{\alpha+\beta-1}$. We note, more generally, that such an integral can be expressed in terms of Gauss's hypergeometric functions $_2F_1$. For example, we take {\bf 7.621} (4) of \cite{4}
\begin{eqnarray}
\int_0^{+\infty}e^{-ht}t^{b-1}{_1}F_1(a;c;kt)\,dt =
  \left\{
    \begin{array}{ll}
      \frac{\Gamma(b)}{h^{b}}\,{_2}F_1\left(a,b;c;\frac{k}{h}\right), & \hbox{$|h|>|k|$;} \nonumber \\
      \frac{\Gamma(b)}{(h-k)^{b}}\,{_2}F_1\left(c-a,b;c;\frac{k}{k-h}\right), & \hbox{$|h-k|>|k|$.} \nonumber
    \end{array}
  \right.
\end{eqnarray}
and we will discuss both cases by analyzing the behavior at $\lambda\rightarrow\pm\infty$ in order to satisfy the condition (3) of the introduction, i.e.,
\begin{eqnarray}\label{2.16}
\lim_{\lambda\rightarrow\pm\infty}\mathcal F^{(\alpha,\beta)}_m(\lambda)=0.
\end{eqnarray}
\paragraph{Case 1.} Substituting $a,b,c,h$ and $k$ by their appropriate parameters in \eqref{2.15}, (i.e., $a=\beta+M,b=\alpha+\beta,c=\beta,h=1$ and $k=-\Lambda$), we find the following transform
\begin{eqnarray}\label{2.17}
\mathcal F^{(\alpha,\beta)}_m(\lambda) &=& 2\pi(-i)^m m!\,\frac{\Gamma(\alpha+\beta)}{\Gamma(\beta)}\,\Lambda^\beta \nonumber \\
  &&\times\sum_{M\in\sigma}\frac{(-1)^M\,\Gamma(\beta+M)}{\prod_{\nu=1}^m i_\nu!\,(\nu!)^{i_\nu}}\,{_2}F_1(\beta+M,\alpha+\beta;\beta;-\Lambda),
\end{eqnarray}
with the restriction $1>|\Lambda|=e^\lambda>0$, ($\lambda<0$), and the passage to the limit $\lambda\rightarrow\pm\infty$
\[
\lim_{\lambda\rightarrow\pm\infty}\mathcal F^{(\alpha,\beta)}_m(\lambda)\quad\rightsquigarrow\quad
\lim_{\lambda\rightarrow\pm\infty} e^{\lambda\beta}{_2}F_1(\beta+M,\alpha+\beta;\beta;-e^{\lambda})=0
\]
provides justification that, for $\beta>0$, \eqref{2.16} holds only for {\it negative} $\lambda$, (i.e., $\lambda\rightarrow -\infty$).\\
\indent Now using the Pfaff's transformation of ${_2}F_1$ \cite{4} (see for instance {\bf 9.131} (1) therein), namely
\[
{_2}F_1\left(A,B;C;z\right)=(1-z)^{-B}\,{_2}F_1\left(B,C-A;C;\frac{z}{z-1}\right),\qquad(|z|<1)
\]
where here $|z|=e^\lambda<1$ and keeping in mind that the exchange between the first-two parameters keeps the hypergeometric function unchanged, then \eqref{2.17} becomes
\begin{eqnarray}\label{2.18}
\mathcal F^{(\alpha,\beta)}_m(\lambda)
   & =&  2\pi\,(-i)^m\,m!\,\frac{\Gamma(\alpha+\beta)}{\Gamma(\beta)}\frac{\Lambda^\beta}{(1+\Lambda)^{\alpha+\beta}} \nonumber \\
   && \times\sum_{M\in\sigma}\frac{(-1)^M\,\Gamma(\beta+M)}{\prod_{\nu=1}^m i_\nu!\,(\nu!)^{i_\nu}}\,{_2}F_1\left(-M,\alpha+\beta;\beta;\frac{\Lambda}{1+\Lambda}\right),
\end{eqnarray}
\paragraph{Case 2.} In this second case, with the same parameters $a,b,c,h$ and $k$ as above, we have
\begin{eqnarray}\label{2.19}
\mathcal F^{(\alpha,\beta)}_m(\lambda)
   & =&  2\pi\,(-i)^m\,m!\,\frac{\Gamma(\alpha+\beta)}{\Gamma(\beta)}\frac{\Lambda^\beta}{(1+\Lambda)^{\alpha+\beta}} \nonumber \\
   && \times\sum_{M\in\sigma}\frac{(-1)^M\,\Gamma(\beta+M)}{\prod_{\nu=1}^m i_\nu!\,(\nu!)^{i_\nu}}\,{_2}F_1\left(-M,\alpha+\beta;\beta;\frac{\Lambda}{1+\Lambda}\right),
\end{eqnarray}
which is identical to the result found in \eqref{2.18},
with the restriction $|1+\Lambda|>|\Lambda|$, (i.e., $\forall\lambda\in\mathbb R$). Thus from \eqref{2.16}, we obtain
\begin{equation*}
\begin{split}
  \lim_{\lambda\rightarrow\pm\infty}\mathcal F^{(\alpha,\beta)}_m(\lambda)
  \quad&\rightsquigarrow\quad \lim_{\lambda\rightarrow\pm\infty}\frac{e^{\lambda\beta}}{(1+e^{\lambda})^{\alpha+\beta}}\,
  {_2}F_1\left(-M,\alpha+\beta;\beta;\frac{e^{\lambda}}{1+e^{\lambda}}\right) \\
  &=\quad\left\{
    \begin{array}{ll}
      \frac{\Gamma(\beta)\,\Gamma(M-\alpha)}{\Gamma(\beta+M)\,\Gamma(-\alpha)}
            \lim_{\lambda\rightarrow+\infty}e^{-\lambda\alpha}=0,                   &\qquad \hbox{(for $\lambda>0$),} \nonumber \\
      \lim_{\lambda\rightarrow-\infty}\frac{e^{\lambda\beta}}{(1+e^{\lambda})^{\alpha+\beta}}=0,   &\qquad \hbox{(for $\lambda<0$),} \nonumber
    \end{array}
  \right.
 \end{split}
 \end{equation*}
and we conclude that \eqref{2.16} is fulfilled if and only if $\mathbb R\ni(\alpha,\beta)>0$ and $\alpha\neq0,1,2,\cdots$ and holds for both {\it positive and negative} $\lambda$, (i.e., $\lambda\rightarrow\pm\infty$). We note that the ratio $\frac{\Gamma(\beta)\,\Gamma(M-\alpha)}{\Gamma(\beta+M)\,\Gamma(-\alpha)}$ in the expression for the case $\lambda>0$ follows by means of {\bf 9.122} (1) of \cite{4}, while in the case $\lambda<0$ we have used the well-known identity ${_2}F_1\left(-M,\alpha+\beta;\beta;0\right)=1$. In summary, it can be seen that both cases {\it lead to the same result}. The special summation in \eqref{2.19} ($\equiv$ \eqref{2.18}) converges, since the hypergeometric function ${_2}F_1$ terminates and its argument belongs to the unit circle $0<\zeta=\frac{e^{\lambda}}{1+e^{\lambda}}<1$, and for $\lambda\rightarrow+\infty$ the argument $\zeta=\frac{e^{\lambda}}{1+e^{\lambda}}\rightarrow1$ is considered as a branch point at infinity, with $M>\alpha$.\\
\indent Thus, by \eqref{2.19}, we complete the proof of theorem.
\end{proof}

\section{Calculations of some special cases}%

\noindent In this section, We illustrate the applications of theorem to a variety of special cases. Many unknown integrals follow directly from \eqref{2.3} by means of the parameters $\alpha,\beta,m$ and $\lambda$ and can be considered as mostly new. We class the results in two categories: $\alpha=\beta$ and $\alpha\neq\beta$ and they are collected here for use in further applications.

\subsection{$\alpha=\beta\neq0$, $m\neq0$ and $\lambda\neq0$}%

Using both Gauss-Legendre duplication formula and Schwarz's reflection principle for gamma functions \cite{1,4}, we get
\begin{eqnarray}\label{3.1}
\int_{-\infty}^{+\infty}e^{-i\lambda s}s^m|\Gamma(\beta+i s)|^2\,ds
     &=& \sqrt{\pi}\,(-i)^m\, m!\left(\frac{2}{1+\cosh\lambda}\right)^\beta\Gamma\left(\beta+\frac{1}{2}\right) \nonumber \\
     &&\times\sum_{M\in\sigma}\frac{(-1)^M\,\Gamma(\beta+M)}{\prod_{\nu=1}^m i_\nu!\,(\nu!)^{i_\nu}}\,
     {_2}F_1\left(-M,2\beta;\beta;\frac{e^\lambda}{1+e^\lambda}\right),
\end{eqnarray}
where $\beta>0$. Consequently, the following formulas follow
\begin{itemize}
\item $\alpha=\beta\neq0$, $m\neq0$ and $\lambda=0$:
 \begin{eqnarray}\label{3.2}
  \int_{-\infty}^{+\infty}s^m|\Gamma(\beta+i s)|^2\,d s
   &=& \sqrt{\pi}\,(-i)^m\, m!\,\Gamma\left(\beta+\frac{1}{2}\right)\nonumber \\
   &&\times\sum_{M\in\sigma}\frac{(-1)^M\,\Gamma(\beta+M)}{\prod_{\nu=1}^m i_\nu!\,(\nu!)^{i_\nu}}\,{_2}F_1\left(-M,2\beta;\beta;\frac{1}{2}\right).
  \end{eqnarray}
\item $\alpha=\beta\neq0$, $m=0$ and $\lambda\neq0$:
 \begin{eqnarray}\label{3.3}
   \int_{-\infty}^{+\infty}e^{-i\lambda s}|\Gamma(\beta+i s)|^2\,d s = \pi\,2^{1-\beta}(1+\cosh\lambda)^{-\beta}\,\Gamma(2\beta).
 \end{eqnarray}
\item $\alpha=\beta=0$, $m\neq0$ and $\lambda\neq0$: We start first by taking the limit $\alpha\rightarrow0^+$, followed by the limit $\beta\rightarrow0^+$, we find
 \begin{equation}\label{3.4}
\int_{-\infty}^{+\infty}e^{-i\lambda s}s^m|\Gamma(i s)|^2\,d s
   = 2\pi\,(-i)^m\,m!\sum_{M\in\sigma}\frac{(-1)^M\,\Gamma(M)}{\prod_{\nu=1}^m i_\nu!\,(\nu!)^{i_\nu}}\,(1+e^\lambda)^{-M}.
 \end{equation}
   \item $\alpha=\beta\neq0$ and $m=\lambda=0$: Using \eqref{3.2} or \eqref{3.3}, we get
 \begin{eqnarray}\label{3.5}
   \int_{-\infty}^{+\infty}|\Gamma(\beta+i s)|^2\,d s = \pi\,2^{1-2\beta}\,\Gamma(2\beta),
 \end{eqnarray}
   which is identical to the result found in Titchmarsh's book using Mellin's transform \cite{8} (see for example, formula 7.8.2, p.~193 therein).
   \item $\alpha=\beta=\lambda=0$ and $m\neq0$: From \eqref{3.4}
 \begin{eqnarray}\label{3.6}
   \int_{-\infty}^{+\infty}s^m|\Gamma(i s)|^2\,d s
       = 2\pi(-i)^m m!\sum_{M\in\sigma}\frac{(-\frac{1}{2})^M\,\Gamma(M)}{\prod_{\nu=1}^m i_\nu!\,(\nu!)^{i_\nu}}.
 \end{eqnarray}
\end{itemize}

\subsection{$\alpha\neq\beta$ and $\alpha\neq0$, $\beta\neq0$}%

In this case, the following possibilities are deduced from \eqref{2.3}:
\begin{itemize}
\item $m\neq0$ and $\lambda=0$:
 \begin{equation}\label{3.7}
 \begin{split}
  \int_{-\infty}^{+\infty}s^m\Gamma(\alpha-i s)\Gamma(\beta+i s)&\,ds
    = \pi(-i)^m m!\,2^{1-\alpha-\beta}\,\frac{\Gamma(\alpha+\beta)}{\Gamma(\beta)} \\
   &\times\sum_{M\in\sigma}\frac{(-1)^M\,\Gamma(\beta+M)}{\prod_{\nu=1}^m i_\nu!\,(\nu!)^{i_\nu}}\,{_2}F_1\left(-M,\alpha+\beta;\beta;\frac{1}{2}\right),
  \end{split}
  \end{equation}
 which is the desired result that we are looking for in \eqref{1.1}. Further physical applications of \eqref{3.7} can be seen in section 5.
\item $m=0$ and $\lambda\neq0$:
 \begin{eqnarray}\label{3.8}
  \int_{-\infty}^{+\infty}e^{-i\lambda s}\Gamma(\alpha-i s)\Gamma(\beta+i s)\,d s = 2\pi\,\frac{e^{\lambda\beta}}{(1+e^\lambda)^{\alpha+\beta}}\,\Gamma(\alpha+\beta).
 \end{eqnarray}
\item $m=\lambda=0$: Using \eqref{3.7} or \eqref{3.8}, we get
 \begin{eqnarray}\label{3.9}
  \int_{-\infty}^{+\infty}\Gamma(\alpha-i s)\Gamma(\beta+i s)\,d s = \pi\,2^{1-\alpha-\beta}\,\Gamma(\alpha+\beta),
 \end{eqnarray}
 which can be regarded as a {\it complex companion integral} to {\bf 6.411} of \cite{4}. Setting $\alpha=0$, there are other relations and here below we mention three more:
\item $\alpha=\lambda=0$ and $m\neq0$: Using \eqref{3.7}, we obtain
 \begin{equation}\label{3.10}
\int_{-\infty}^{+\infty}s^m\Gamma(-i s)\Gamma(\beta+i s)\,ds
   = \pi(-i)^m m!\,2^{1-\beta}\sum_{M\in\sigma}\frac{(-\frac{1}{2})^M\,\Gamma(\beta+M)}{\prod_{\nu=1}^m i_\nu!\,(\nu!)^{i_\nu}},
 \end{equation}
 where for $\beta=0$, we obtain once again \eqref{3.6}.
\item $\alpha=m=0$ and $\lambda\neq0$: From \eqref{3.8}, we have
 \begin{equation}\label{3.11}
  \int_{-\infty}^{+\infty}e^{-i\lambda s}\Gamma(-i s)\Gamma(\beta+i s)\,d s = 2\pi\left(\frac{e^\lambda}{1+e^\lambda}\right)^\beta\Gamma(\beta).
 \end{equation}
\item $\alpha=m=\lambda=0$: Using \eqref{3.9} or \eqref{3.11}, we get
 \begin{equation}\label{3.12}
  \int_{-\infty}^{+\infty}\Gamma(-is)\Gamma(\beta+is)\,ds = \pi 2^{1-\beta}\Gamma(\beta).
 \end{equation}
\end{itemize}

\indent The identification \eqref{3.12}, for example, provides a very useful {\it integral representation} of the gamma function in terms of two other complex gamma functions. For example, a direct application of \eqref{3.12} for $\beta=1$ gives $\pi$, while the software system {\sc Mathematica} gives an integral of ${\rm csch}(\pi s)$ which does not converge on $\left]-\infty,+\infty\right[$.

\section{Mathematical applications of the theorem: Analysis and number theory}%

The main aim of this section is to show the efficiency of our formula \eqref{2.3} in establishing some well known, as well as unknown, results. It should be noted that applications of our formula are not limited to only the presented results but further identities can be obtained. It is obvious that we need the use of Table~1 to evaluate the different integrals.
\subsection{Test 1: some known integrals}%
For the first test, we try to recover some well known results using the relations deduced in section 3.
\subsubsection{Example 1: $\alpha=\beta=\lambda=0$ and $m=2$.}%
We use \eqref{3.6} and we get
\begin{eqnarray}\label{4.1}
  \mathcal F_2^{(0,0)}(0) &=& \int_{-\infty}^{+\infty}s^2\,|\Gamma(i s)|^2\,d s\nonumber \\
   &=& 2\pi\,(-i)^2\,2!\left(\frac{\left(-\frac{1}{2}\right)^2}{2!\,0!}\frac{\Gamma(2)}{(1!)^2(2!)^0}+
                                \frac{\left(-\frac{1}{2}\right)^1}{0!\,1!}\frac{\Gamma(1)}{(1!)^0(2!)^1}\right)\nonumber \\
   &=& \frac{\pi}{2}.
\end{eqnarray}
\indent On the other hand combining the integral $\int_0^{+\infty}\frac{s\,d s}{\sinh\pi s}=\frac{1}{4}$ in {\bf 3.521} (1) with the identity $|\Gamma(i s)|^2=\pi/(s\,\sinh\pi s)$ in {\bf 8.332} (1) of \cite{4}, taking into account that the integrand is an even function, namely
\[
   \int_0^{+\infty}\frac{s}{\sinh\pi s}\,d s=\frac{1}{\pi}\int_0^{+\infty} s^2\frac{\pi}{s\,\sinh\pi s}\,d s
    =\frac{1}{2\pi}\int_{-\infty}^{+\infty} s^2\,|\Gamma(i s)|^2\,d s=\frac{1}{4},
\]
we recover the result obtained in \eqref{4.1}.
\subsubsection{Example 2: $\alpha=\beta=\frac{1}{2}$, $m=4$ and $\lambda=0.$}%
Using Table~1 for $m=4$ with the help of \eqref{3.2}, we obtain
\begin{eqnarray}\label{4.2}
\mathcal F_4^{(1/2,1/2)}(0) &=& \int_{-\infty}^{+\infty}s^4\,|\Gamma\left(1/2+i s\right)|^2\,ds \nonumber \\
    &=& \sqrt\pi\, 4!\,\Bigg(\frac{(-1)^4}{4!}\frac{\Gamma(\frac{9}{2})}{1!}\,_2F_1(-4,1;1/2;1/2)
    +\frac{(-1)^3}{2!}\frac{\Gamma(\frac{7}{2})}{2!}\,_2F_1(-3,1;1/2;1/2) \nonumber \\
    && +\frac{(-1)^2}{1!}\frac{\Gamma(\frac{5}{2})}{3!}\,_2F_1(-2,1;1/2;1/2)
    +\frac{(-1)^2}{2!}\frac{\Gamma(\frac{5}{2})}{(2!)^2}\,_2F_1(-2,1;1/2;1/2) \nonumber \\
    && +\frac{(-1)^1}{1!}\frac{\Gamma(\frac{3}{2})}{4!}\,_2F_1(-1,1;1/2;1/2)\Bigg) \nonumber \\
    &=& \frac{5\pi}{16}.
\end{eqnarray}
\indent As in the example 1, we combine the integral $\int_0^{+\infty}\frac{s^4\,d s}{\cosh\pi s}=\frac{5}{32}$ given by {\bf 3.523} (4) with the identity $|\Gamma(1/2+i s)|^2=\pi/\cosh\pi s$ in {\bf 8.332} (2) of \cite{4}, namely
\begin{equation*}
\int_0^{+\infty}\frac{s^4}{\cosh\pi s}\,d s=\frac{1}{2\pi}\int_{-\infty}^{+\infty} s^4\frac{\pi}{\cosh\pi s}
\,d s=\frac{1}{2\pi}\int_{-\infty}^{+\infty} s^4\,|\Gamma\left(1/2+i s\right)|^2\,d s=\frac{5}{32},
\end{equation*}
so, we obtain the result \eqref{4.2}.
\subsubsection{Example 3: $\alpha=\beta=1$ and $m=\lambda=0.$}%
In this case, we can use \eqref{3.5} and obtain
\begin{eqnarray}\label{4.3}
  \mathcal F_0^{(1,1)}(0)&=&\int_{-\infty}^{+\infty}|\Gamma(1+i s)|^2\,d s = \pi\,2^{-1}\,\Gamma(2) = \frac{\pi}{2}.
\end{eqnarray}
\indent Once again, if we combine the integral $\int_0^{+\infty}\frac{s\,d s}{\sinh\pi s}=\frac{1}{4}$ of the example 1 with the identity $|\Gamma(1+i s)|^2=\pi s/\sinh\pi s$ given by {\bf 8.332} (3) in \cite{4}, we get
\[
\int_0^{+\infty}\frac{s}{\sinh\pi s}\,d s=\frac{1}{\pi}\int_0^{+\infty}\frac{\pi s}{\sinh\pi s}\,d s
=\frac{1}{2\pi}\int_{-\infty}^{+\infty}|\Gamma(1+i s)|^2\,d s=\frac{1}{4},
\]
which leads to the same result as in \eqref{4.3} and \eqref{4.1}.
\subsection{Test 2: Euler polynomials and some related identities}%
Euler's polynomials, $E_n(x)$, have been defined in a number of different ways (\cite{4}, section {\bf 9.6}) and occur frequently in analysis and number theory. This second test aims to generate another and new representation of these polynomials, different from those known in the literature. As a consequence, some unknown results can be deduced by choosing suitably the parameters. Let us start by given the following lemma:
\newtheorem{Lemma}{Lemma}[section]%
\begin{Lemma} For $0<\beta<1$ and $m\geq0$, the Euler polynomials are defined otherwise by
\begin{eqnarray}\label{4.4}
   E_m(\beta) &=& \frac{(-1)^m\,m!}{\Gamma(\beta)}\,\sum_{M\in\sigma}\frac{(-1)^M\,\Gamma(\beta+M)}{\prod_{\nu=1}^m i_\nu!\,(\nu!)^{i_\nu}}\,{_2}F_1\left(-M,1;\beta;\frac{1}{2}\right).
\end{eqnarray}
\end{Lemma}
Before indicating a proof of Lemma 4.1, it is important to mention the relationship between $E_m(\beta)$'s and a discrete summation on the product $(\beta)_M\times{_2}F_1\left(-M,1;\beta;\frac{1}{2}\right)$, where $(\beta)_M=\Gamma(\beta+M)/\Gamma(\beta)$ is the Pochhammer's symbol.
\begin{proof}
From the identity \eqref{3.8}, we take the $m$th derivative with respect to $\lambda$ and we obtain
\[
\int_{-\infty}^{+\infty}e^{-i\lambda s} s^m\Gamma(\alpha-i s)\Gamma(\beta+i s)\,d s=2\pi\,i^m\,\Gamma(\alpha+\beta)\,
\frac{d^m}{d\lambda^m}\left[\frac{e^{\lambda\beta}}{(1+e^\lambda)^{\alpha+\beta}}\right],
\]
which is exactly our previously formula \eqref{2.4} arising from the theorem, i.e.,
\begin{eqnarray}\label{4.5}
  \frac{d^m}{d\lambda^m}\left[\frac{e^{\lambda\beta}}{(1+e^\lambda)^{\alpha+\beta}}\right] &=& \frac{(-1)^m\,m!}{\Gamma(\beta)}\,\frac{e^{\lambda\beta}}{(1+e^\lambda)^{\alpha+\beta}} \nonumber \\
  &&\times\sum_{M\in\sigma}\frac{(-1)^M\,\Gamma(\beta+M)}{\prod_{\nu=1}^m i_\nu!\,(\nu!)^{i_\nu}}\,{_2}F_1\left(-M,\alpha+\beta;\beta;\frac{e^\lambda}{1+e^\lambda}\right).
\end{eqnarray}
\indent Now setting $\alpha+\beta=1$ in \eqref{4.5}, (i.e., for $\alpha=1-\beta>0$ we get $0<\beta<1$), we can observe the similarity between \eqref{4.5} and the generating-function for the Euler polynomials \cite{4},
\[
\frac{2\,e^{\lambda\beta}}{1+e^\lambda}=\sum_{n=0}^{+\infty}E_n(\beta)\,\frac{\lambda^n}{n!}, \qquad(|\lambda|<\pi).
\]

\indent Then, taking the $m$th derivative of the last identity
\[
\frac{d^m}{d\lambda^m}\left(\frac{2\,e^{\lambda\beta}}{1+e^\lambda}\right)
=\sum_{n=m}^{+\infty}E_n(\beta)\,\frac{\lambda^{n-m}}{(n-m)!},\qquad(n\geq m),
\]
we may identify $E_n(\beta)$'s as successive derivatives, for $n=m$, of the generating-function by setting the limit $\lambda\rightarrow0$. Then from \eqref{4.5} we have
\begin{equation*}
\begin{split}
E_m(\beta) &= \frac{d^m}{d\lambda^m}\left(\frac{2\,e^{\lambda\beta}}{1+e^\lambda}\right)\Bigg|_{\lambda=0}\\
           &= \frac{(-1)^m\,m!}{\Gamma(\beta)}\,\sum_{M\in\sigma}\frac{(-1)^M\,\Gamma(\beta+M)}{\prod_{\nu=1}^m
                 i_\nu!\,(\nu!)^{i_\nu}}\,{_2}F_1\left(-M,1;\beta;\frac{1}{2}\right).
\end{split}
\end{equation*}
\indent This proves \eqref{4.4} in the Lemma 4.1.
\end{proof}

\indent We can also apply Lemma 4.1 to obtain some new {\it unknown} results concerning, in particular, the calculations of: (i) Euler's $E_m$ and (ii) Bernoulli's $B_m$ numbers, (iii) monomial expansions $\beta^m$ as generating functions, (iv) residues of gamma function, and (v) the associated Laguerre's polynomials $L_m^{(-m)}(\beta)$ in terms of a discrete summation.\\
\indent To this end, we have the following lemma:
\begin{Lemma} We have
\begin{enumerate}
  \item Euler's numbers:
  \begin{equation}\label{4.6}
      E_m = \frac{(-2)^m\,m!}{\sqrt\pi}\,\sum_{M\in\sigma}\frac{(-1)^M\,\Gamma\left(M+\frac{1}{2}\right)}{\prod_{\nu=1}^m
                 i_\nu!\,(\nu!)^{i_\nu}}\,{_2}F_1\left(-M,1;\frac{1}{2};\frac{1}{2}\right),\,\,\quad(m\geq0).
  \end{equation}
  \item Bernoulli's numbers:
  \begin{eqnarray}
      B_{m+1} &=& \frac{(-1)^{m+1}\,(m+1)\,m!}{2^{m+1}-1}\,\sum_{M\in\sigma}\frac{(-\frac{1}{2})^{M+1}\,\Gamma(M+1)}{\prod_{\nu=1}^m
             i_\nu!\,(\nu!)^{i_\nu}},\,\,\qquad(m\geq1),\label{4.7} \\
      B_{m}   &=&  \frac{(-1)^{m+1}\,m\,m!}{2^{m}-1}\,\sum_{M\in\sigma}\frac{(-\frac{1}{2})^{M}\,\Gamma(M)}{\prod_{\nu=1}^m
             i_\nu!\,(\nu!)^{i_\nu}},\qquad\qquad\qquad\quad(m\geq1).\label{4.8}
  \end{eqnarray}
  \item Monomial expansions (generating functions):
  \begin{equation}\label{4.9}
  \beta^m = \frac{(-1)^{m}\,m!}{\Gamma(\beta)}\,\sum_{M\in\sigma}\frac{(-1)^M\,\Gamma\left(\beta+M\right)}{\prod_{\nu=1}^m
                 i_\nu!\,(\nu!)^{i_\nu}},\,\,\,\,\,\qquad\quad(\mathbb R\ni\beta\neq-m,\, m\geq0).
  \end{equation}
  \item Residues of gamma function at $M=-m$:
  \begin{equation}\label{4.10}
  {\rm Res}\{-m,\Gamma\}\stackrel{\mbox{\tiny{\rm (def)}}}{\equiv}\frac{(-1)^m}{m!} = \sum_{M\in\sigma}\frac{(-1)^M\,\Gamma\left(M+1\right)}{\prod_{\nu=1}^m
                 i_\nu!\,(\nu!)^{i_\nu}},\qquad\qquad\,\,\,\quad(m\geq0).
  \end{equation}
  \item Associated Laguerre polynomials
  \begin{equation}\label{4.11}
  L_m^{(-m)}(\beta) \stackrel{\mbox{\tiny{\rm (def)}}}{\equiv} \frac{(-1)^m}{m!}\,\beta^m = \frac{1}{\Gamma(\beta)}\,\sum_{M\in\sigma}\frac{(-1)^M\,\Gamma\left(\beta+M\right)}{\prod_{\nu=1}^m
                 i_\nu!\,(\nu!)^{i_\nu}},\,\,\,\qquad(m\geq0).
  \end{equation}
\end{enumerate}
\end{Lemma}

\begin{proof}
\begin{enumerate}
  \item We can prove \eqref{4.6} by setting the definition of the Euler numbers, i.e., $E_m=2^m E_m\left(\frac{1}{2}\right),$ given by {\bf 9.655} (3) in \cite{4}, for the value $\beta=\frac{1}{2}$.
  \item The identities \eqref{4.7} and \eqref{4.8} can be calculated using {\bf 9.655} (2) of \cite{4}, namely
        \[ E_m(1)=2\frac{2^{m+1}-1}{m+1}\,B_{m+1}, \]
        by setting $\beta=1$ in \eqref{4.4}. It should be noted that an alternative definition in \eqref{4.8} for the Bernoulli numbers must be taken into account since \eqref{4.7} starts by computing the value of $B_2=\frac{1}{6}$, while \eqref{4.8} gives in first the value of $B_1=-\frac{1}{2}$. Indeed a simple application of \eqref{4.7} leads for $m=M=1$ to: $B_2=\frac{(-1)^{1+1}(1+1)1!}{2^{1+1}-1}\frac{(-1/2)^{1+1}\Gamma(1+1)}{1!(1!)^1}\equiv\frac{1}{6}$, while for the same value of $m$, \eqref{4.8} gives: $B_1=\frac{(-1)^{1+1}1\cdot1!}{2^1-1}\frac{(-1/2)^{1}\Gamma(1)}{1!(1!)^1}\equiv-\frac{1}{2}$, as expected.
  \item The monomial expansion \eqref{4.9} can be obtained by setting $\alpha+\beta=0$ in \eqref{4.5} (with $\beta>0$), followed by evaluating the $m$-derivative of $e^{\lambda\beta}$ and taking the limit $\lambda\rightarrow0$ to both the sides of equality. We can see \eqref{4.9} as the generating function of the monomial $\beta^m$ like that involves the Stirling number of the second kind $\mathfrak S_m^{(M)}$ (see, for example {\bf 9.745} (1) in \cite{4}) which allows us to count the number of ways of partitioning a set of $m$ elements into $M$ subsets.
  \item From \eqref{4.9} the reader can easily obtain \eqref{4.10}, by setting $\beta=1$.
  \item Finally, following the definition of the associated Laguerre polynomials \cite{4} (see, {\bf 8.973} (4)), the expression in \eqref{4.11} can be obtained from \eqref{4.9}.
\end{enumerate}
\end{proof}

\indent As an application, we verified the correctness of all expressions of \eqref{4.4} and \eqref{4.6}-\eqref{4.11}, for even and odd $m$ up to 8, by recovering the standard results (i.e., all polynomials, monomials and numbers). The reader can easily check that for the case $m=6$ corresponds 11 different combinations of $i_k$, the case $m=7$ has 15 possibilities and the case $m=8$ accepts 22 possibilities, therefore we leave it as an exercise to the interested reader.

\section{Physical applications of the theorem: Quantum mechanics in phase-space}%

In the last decade, there has been great interest to study quantum systems endowed with position-dependent effective mass (PDEM) for Schr\"odinger equation (see, \cite{12,13,14} and references therein). In our paper \cite{3} we succeeded to construct, analytically and numerically, the Wigner's distribution functions (WDF) for the generalized Laguerre polynomials using to that end an exponentially decaying mass function.\\
\indent In this context, we observed that an important application of our theorem, \eqref{2.3}, arises when we deal with computations of expectation values for position and momentum operators using the Weyl's transforms (WT).\\
\indent Briefly, given the eigenfunctions $\psi_n(x)$ corresponding to the energy eigenvalues $E_n$, WDF in the phase space is defined by (see, e.g., the report papers \cite{15,16})
\begin{eqnarray}\label{5.1}
  \mathcal W(\psi_n|x,p)&=&\frac{1}{2\pi}\int_{-\infty}^{+\infty}e^{-i py}\,\psi^\ast_n\left(x-\frac{y}{2}\right)\psi_n\left(x+\frac{y}{2}\right)\,d y,
\end{eqnarray}
with $\hbar=1$. Using \eqref{5.1} the expectation value of an operator $\hat{\mathcal O}$ is given through
\begin{eqnarray}\label{5.2}
  \langle{\mathcal O}\rangle&=&\int_{-\infty}^{+\infty}dx\int_{-\infty}^{+\infty}dp\,\,\mathcal W(\psi_n|x,p)\,\mathfrak{W}[\hat{\mathcal O}],
\end{eqnarray}
where $\mathfrak{W}[\hat{\mathcal O}]$ is the central object in the phase space (or deformed-) quantization, called the Weyl transformation of the operator $\hat{\mathcal O}$ and defined by
\begin{eqnarray}\label{5.3}
  \mathfrak{W}[\hat{\mathcal O}]&=&\int_{-\infty}^{+\infty}e^{-i py}\,\left\langle x+\frac{y}{2}\left|\hat{\mathcal O}\right|x-\frac{y}{2}\right\rangle\,d y.
\end{eqnarray}

\subsection{Generalized Laguerre solutions and their Wigner's functions}%

Using the coordinate transformations, the eigenfunctions of many quantum systems can be expressed in terms of orthogonal polynomials, amongst them we can recover a class of quantum potentials deduced from the generalized Laguerre polynomials \cite{17}.\\
\indent As an example of application, we will focus our attention on the three-dimensional harmonic oscillator where the effective potential $V_{\rm eff}(x)$, the energy eigenvalues $E_n$ and the normalized eigenfunctions $\psi_n(x)$ are given by
\begin{eqnarray}
  V_{\rm eff}(x) &=& -\left(l+\frac{3}{2}\right)+\frac{1}{2}\mu^2(x)+\frac{l(l+1)}{2\,\mu^2(x)}
                 +\frac{1}{8m(x)}\left[\frac{m''(x)}{m(x)}-\frac{7}{4}\left(\frac{m'(x)}{m(x)}\right)^2\right],\label{5.4} \\
  E_n &=& 2n,\label{5.5} \\
  \psi_n(x) &=& \sqrt{\frac{a\,n!}{\Gamma\left(n+l+\frac{3}{2}\right)}}\,m^{1/4}(x)\,\mu^{l+1}(x)\,
                e^{-\mu^2(x)/2}\,L_n^{(l+1/2)}\left(\mu^2(x)\right),\label{5.6}
\end{eqnarray}
where $l$ is the angular momentum quantum number $\left(l\neq-\frac{3}{2},-\frac{5}{2},-\frac{7}{2},\cdots\right)$ and $L_n^{(s)}(\cdot)$ are the generalized Laguerre polynomials. The profile of the mass function is given by
\[
 m(x)=e^{-a|x|}\quad\Rightarrow\quad\mu(x):= \int^x\sqrt{m(\eta)}\,d\eta=
                                           \left\{
                                             \begin{array}{ll}
                                               -\frac{2}{a}\,e^{-a x/2}, & \hbox{($x > 0)$;} \\
                                               +\frac{2}{a}\,e^{+a x/2}, & \hbox{($x < 0)$.}
                                             \end{array}
                                           \right.
\]
with $a(\neq0)$ being the inverse quantum-well width. It is obvious that the case $a=0$ reduces the system to a constant mass, i.e., $m(x) = 1$.\\
\indent In the sequel, we will frequently work under the special choice of canonical transformations $x\rightarrow\mu(x)$ and $p\rightarrow\pi(x,p)=p/\sqrt{m(x)}$ adapted for better describing a problem endowed with PDEM. Then using \eqref{5.1}, WDF corresponding to the eigenfunctions \eqref{5.6} is given by \cite{3}
\begin{eqnarray}\label{5.7}
  \mathcal W\left(\psi_n\big|x,p\right) &=& \frac{2\,n!}{\pi\,\Gamma\left(n+l+\frac{3}{2}\right)}\sum_{l_1=0}^n\sum_{l_2=0}^n\frac{(-1)^{l_1+l_2}}{l_1!\,l_2!}{n+l+\frac{1}{2}\choose n-l_1}{n+l+\frac{1}{2}\choose n-l_2}\nonumber \\
                                     &&\times\left(\mu^2(x)\right)^{l+l_1+l_2+3/2}\,K_{l_1-l_2-2ip/a}\left(\mu^2(x)\right).
\end{eqnarray}
where $K_\delta(\cdot)$ are the modified Bessel functions of the third kind (MacDonald's functions).
\subsection{Applications of \eqref{3.7}: computations of expectation values}%
It is known that the Weyl transformation (WT), defined in \eqref{5.3}, quantizes classical coordinates $(\mu,\pi)$ to its corresponding quantum operators $(\hat\mu,\hat\pi)$ \cite{15,16}. Quantum canonical transformations, which preserve the Dirac brackets, are regarded as suitable transforms to find such a correspondence and in our case they are given by \cite{3,18}
\[
\hat x\rightarrow\mu(\hat x),\qquad{\rm and}\qquad \hat p\rightarrow\pi(\hat x,\hat p)=\frac{1}{\sqrt{m(\hat x)}}\,\hat p.
\]

\indent Before computing the appropriate expectation values $\langle\mu^q\rangle$ and $\langle\pi^q\rangle$, ($q=1,2$), for operators $\mu^q(\hat x)$ and $\pi^q(\hat x,\hat p)$, we first calculate the Weyl transformations $\mathfrak W[\mu^q(\hat x)]$ and $\mathfrak W[\pi^q(\hat x,\hat p)]$ using \eqref{5.3} and we find\footnote{See, for instance, the proof of \eqref{5.10} in the appendix A of \cite{3}.}:
\begin{eqnarray}
  \mathfrak W[\mu^q(\hat x)]        &=& \mu^q(x), \label{5.8} \\
  \mathfrak W[\pi(\hat x,\hat p)]   &=& \frac{2}{a}\,\mu^{-1}(x)\,p+\frac{i}{2}\,\mu^{-1}(x),\label{5.9} \\
  \mathfrak W[\pi^2(\hat x,\hat p)] &=& \frac{4}{a^2}\,\mu^{-2}(x)\,p^2+\frac{2i}{a}\,\mu^{-2}(x)\,p. \label{5.10}
\end{eqnarray}
\indent Using \eqref{5.7}-\eqref{5.10}, the appropriate expectation values can be computed by means of \eqref{5.2} in terms of the auxiliary mass function $\mu(x)$, which lead to
\begin{eqnarray}
 \langle\mu^q\rangle &\equiv& \frac{1}{a}\,\Sigma_{n,l}\,\int_{-\infty}^{+\infty}d p\int_{0}^{+\infty}
     \kappa^{l+l_1+l_2+\frac{q+1}{2}}K_{l_1-l_2-\frac{2ip}{a}}(\kappa)\,d\kappa,  \label{5.11} \\
 \langle\pi\rangle &\equiv& \frac{2}{a^2}\,\Sigma_{n,l}
            \int_{-\infty}^{+\infty}p\,dp\int_{0}^{+\infty}\kappa^{l+l_1+l_2}K_{l_1-l_2-\frac{2i p}{a}}(\kappa)\,d\kappa\nonumber \\
            &&+\frac{i}{2a}\,\Sigma_{n,l}
                             \int_{-\infty}^{+\infty}d p\int_{0}^{+\infty}\kappa^{l+l_1+l_2}K_{l_1-l_2-\frac{2i p}{a}}(\kappa)\,d\kappa, \label{5.12} \\
 \langle\pi^2\rangle &\equiv& \frac{4}{a^3}\,\Sigma_{n,l}
            \int_{-\infty}^{+\infty}p^2\,d p\int_{0}^{+\infty}\kappa^{l+l_1+l_2-\frac{1}{2}}K_{l_1-l_2-\frac{2\,i p}{a}}(\kappa)\,d\kappa\nonumber \\
                     &&+\frac{2i}{a^2}\,\Sigma_{n,l}
                             \int_{-\infty}^{+\infty}p\,d p\int_{0}^{+\infty}\kappa^{l+l_1+l_2-\frac{1}{2}}K_{l_1-l_2-\frac{2i p}{a}}(\kappa)\,d\kappa, \label{5.13}
\end{eqnarray}
where $\kappa\equiv\mu^2(x)$ and
\begin{equation*}
\begin{split}
\Sigma_{n,l} &= \frac{2}{\pi}\sum_{l_1=0}^n\sum_{l_2=0}^n\gamma_{n,l,l_1,l_2}\\
             &=\frac{2}{\pi}\frac{n!}{\Gamma\left(n+l+\frac{3}{2}\right)}
  \sum_{l_1=0}^n\sum_{l_2=0}^n\frac{(-1)^{l_1+l_2}}{l_1!\,l_2!}\,{n+l+\frac{1}{2}\choose n-l_1}\,{n+l+\frac{1}{2}\choose n-l_2}.
\end{split}
\end{equation*}

\indent Performing $\kappa$-integration in \eqref{5.11}-\eqref{5.13} using the identity {\bf 6.561} (16) of \cite{4}, namely
\[
\int_0^{+\infty}\kappa^\gamma\,K_\delta(\kappa)\,d\kappa =2^{\gamma-1}\Gamma\left(\frac{1+\gamma+\delta}{2}\right)\Gamma\left(\frac{1+\gamma-\delta}{2}\right), \qquad(\Re(\gamma+1\pm\delta)>0)
\]

\noindent we obtain
\begin{eqnarray}
  \langle\mu^q\rangle &=& 2^{\frac{q-1}{2}}\,\Sigma_{n,l}\,2^{l+l_1+l_2}\int_{-\infty}^{+\infty}
     \Gamma\left(\sigma^{(q)}_1-i s\right)\Gamma\left(\sigma^{(q)}_2+i s\right)\,d s, \label{5.14} \\
  \langle\pi\rangle   &=& \Sigma_{n,l}\,2^{l+l_1+l_2}\int_{-\infty}^{+\infty}s\,
     \Gamma\left(\varsigma_1-i s\right)\Gamma\left(\varsigma_2+i s\right)\,d s\nonumber \\
     &&+i\,\Sigma_{n,l}\,2^{l+l_1+l_2-2}\int_{-\infty}^{+\infty}\,
     \Gamma\left(\varsigma_1-i s\right)\Gamma\left(\varsigma_2+i s\right)\,d s,  \label{5.15} \\
  \langle\pi^2\rangle &=& \Sigma_{n,l}\,2^{l+l_1+l_2+\frac{1}{2}}\int_{-\infty}^{+\infty}s^2\,
     \Gamma\left(\zeta_1-i s\right)\Gamma\left(\zeta_2+i s\right)\,d s\nonumber \\
     &&+i\,\Sigma_{n,l}\,2^{l+l_1+l_2-\frac{1}{2}}\int_{-\infty}^{+\infty}s\,
     \Gamma\left(\zeta_1-i s\right)\Gamma\left(\zeta_2+i s\right)\,d s, \label{5.16}
\end{eqnarray}
where $s=p/a$ and the parameters $\sigma^{(q)}_{1,2}$, $\varsigma_{1,2}$ and $\zeta_{1,2}$ are defined in the following
\begin{eqnarray*}
\sigma^{(q)}_{1,2} &=& \frac{l}{2}+l_{1,2}+\frac{q+3}{4}, \qquad(q=1,2) \\
\varsigma_{1,2}    &=& \frac{l}{2}+l_{1,2}+\frac{1}{2},\\
\zeta_{1,2}        &=& \frac{l}{2}+l_{1,2}+\frac{1}{4}.
\end{eqnarray*}
\indent The integrals just mentioned can be directly evaluated by applying, for example, \eqref{3.9} to compute \eqref{5.14} and the second integral of \eqref{5.15} for $m=0$, and we can use \eqref{3.7} to compute the first integrals in \eqref{5.15} and \eqref{5.16} for $m=1,2$. Using Table~1 for odd and even $m$ up to 2, the expectation values are finally given through the expressions \cite{3}
\begin{eqnarray}
  \langle\mu^q\rangle_{n,l} &=& \sum_{l_1=0}^n\sum_{l_2=0}^n \gamma_{n,l,l_1,l_2}\,\Gamma\left(l+l_1+l_2+\frac{q+3}{2}\right)
                             ,\quad(q=1,2)  \label{5.17} \\
  \langle\pi\rangle_{n,l}   &=& -i\sum_{l_1=0}^n\sum_{l_2=0}^n \gamma_{n,l,l_1,l_2}\,\left(l_1-l_2-\frac{1}{2}\right)\,\Gamma\left(l+l_1+l_2+1\right), \label{5.18} \\
  \langle\pi^2\rangle_{n,l} &=& \sum_{l_1=0}^n\sum_{l_2=0}^n \gamma_{n,l,l_1,l_2}\,\left(l+\frac{1}{2}-l_1^2-l_2^2+2l_1l_2+2l_1\right)\Gamma\left(l+l_1+l_2+\frac{1}{2}\right), \label{5.19}
\end{eqnarray}
which depend only on quantum numbers $n$ and $l$.\\
\indent Finally we can use \eqref{5.17}-\eqref{5.19} to evaluate the spread in position and momentum, defined by $(\Delta\Theta)_{n,l}:=\sqrt{\langle\Theta^2\rangle_{n,l}-\langle\Theta\rangle^2_{n,l}}$,\,($\Theta=\mu,\pi$), in order to verify the universality of the Heisenberg uncertainty principle $(\Delta\mu)_{n,l}\cdot(\Delta\pi)_{n,l}$. In this context, we have made interesting observations in a previous paper \cite{3} consisting on the occurrence of a common pattern of the type
\[
\lim_{l\rightarrow\infty}(\Delta\mu)_{n,l}\cdot(\Delta\pi)_{n,l}\sim n+\frac{1}{2},
\]
which

\begin{itemize}
  \item provides for the measurement of {\it observables at lower bound} and suggesting that $\frac{1}{2}<(\Delta\mu)_{0,l}\cdot(\Delta\pi)_{0,l}<\frac{3}{4}$ is almost at its minimum for the ground-state level ($n=0$), and also
  \item establishes the {\it quantum-classical connection} for large values of $l$, ($l\rightarrow\infty$).
\end{itemize}

\section{Conclusion}%

We have shown that the Fourier integral containing a pair of Euler complex gamma functions with a monomial can be used to prove many usual integrals and generate many other unknown identities that are found to be very useful in mathematical as well in physical applications. Hence, we believe that the result in \eqref{2.3} is of some interest for researches in mathematical physics, as well as in other fields of applied science.\\
\indent The readers can remark that the only disadvantage in using \eqref{2.3} may lie in the computation of the special summation, due to Fa\`a di Bruno's formula, which cannot be conveniently used, specially for large values of $m$. This remark is due to the fact that we are constrained to counting all different possible combinations of $i_k$ satisfying equations \eqref{2.4} and \eqref{2.5}, unless using a computer programming to evaluate a such summation for large values of $m$.


\pdfbookmark[1]{References}{ref}
\LastPageEnding

\end{document}